\documentclass[runningheads,a4paper]{llncs}

\usepackage{amssymb}
\usepackage{graphicx}

\usepackage{url}  
\urldef{\mailsa}\path|{alfred.hofmann, brigitte.apfel, ursula.barth, christine.guenther,|
\urldef{\mailsb}\path|ingrid.haas, frank.holzwarth, anna.kramer, leonie.kunz, nicole.sator,|
\urldef{\mailsc}\path|erika.siebert-cole, peter.strasser, lncs}@springer.com|
\newcommand{\keywords}[1]{\par\addvspace\baselineskip
\noindent\keywordname\enspace\ignorespaces#1}

\usepackage{cite}
\usepackage{graphicx}
\usepackage{subfigure}
\RequirePackage{epsfig}
\usepackage{xspace}
\usepackage{amsmath}
\usepackage{amssymb}
\usepackage{amscd}
\newcommand{\PolRing}{\Bbbk[x_{1},\ldots,x_{N}]}%default notation of
\newcommand{\SetOfAllMonomials}{{\mathbb{Z}}_{\geq 0}^{n}}

\newcommand{\eg}{{\em e.g.},\xspace }

\begin{document}

\mainmatter  % start of an individual contribution

\title{On Consensus under Polynomial Protocols}

\titlerunning{Polynomial Consensus}

\author{Joel George Manathara\inst{1} \and Ambedkar Dukkipati\inst{2}
  \and Debasish Ghose\inst{3}
%\thanks{Please note that the LNCS Editorial assumes that all authors have used
%the western naming convention, with given names preceding surnames. This determines
%the structure of the names in the running heads and the author index.}%
}
\authorrunning{J. G. Manathara, A. Dukkipati, D. Ghose}

\institute{Department of Aerospace Engineering\\
Indian Institute of Science, India\\
Email: joel@aero.iisc.ernet.in
\and
Department of Computer Science and Automation\\
Indian Institute of Science, India\\
Email: ambedkar@csa.iisc.ernet.in \\
\and
Department of Aerospace Engineering\\
Indian Institute of Science, India\\
Email: dghose@aero.iisc.ernet.in
}


\toctitle{Lecture Notes in Computer Science}
\tocauthor{Authors' Instructions}
\maketitle

%-----------------------------------------------
\begin{abstract}
\emph{}
In this paper we explore the possibility of using computational
algebraic methods to analyze a class of consensus protocols. We state
some necessary conditions for convergence under consensus protocols
that are polynomials. 
\end{abstract}
%-----------------------------------------------

%===============================================
\section{Introduction}
Consensus among agents is an important problem in multi-agent
systems and congestion/flow control in communication networks. We
formulate the problem as follows. 

Consensus protocol is an ordered set of functions
$(f_1,f_2,\ldots,f_N)$,
where $f_i:\mathbb{R}^N\longrightarrow\mathbb{R},\; i=1,\ldots,N$.
We consider the following first order consensus dynamics among $n$
agents under this protocol
\begin{equation}
\label{eqnset} 
\dot{x}_i= f_i(x_1,x_2,\ldots,x_N) \:\:\: i=1,\ldots,N. 
\end{equation}

%\begin{eqnarray}
%\dot{x}_1 &=& f_1(x_1,x_2,\ldots,x_N) \notag \\
%\dot{x}_2 &=& f_2(x_1,x_2,\ldots,x_N) \notag \\
%& \vdots & \notag \\
%\dot{x}_N &=& f_N(x_1,x_2,\ldots,x_N) \label{eqnset}
%\end{eqnarray}
Here $x_1,x_2,\ldots,x_N$ are the states of agents, the initial
values of which are assumed to be known to the respective agents.
The problem of interest is to determine whether the system of
equations in (\ref{eqnset}) attains a consensus. Consensus is an
equilibrium at which $x_1=x_2=\ldots=x_N$. In other words, we are
interested to know whether there exist an $\alpha\in\mathbb{R}$ such
that when $\mathbf{x}=\mathbf{x}_e$, where
$\mathbf{x}=(x_1,x_2,\ldots,x_N)$, and $\mathbf{x}_e$ is of the form
$\mathbf{x}_e=(1,1,\ldots,1)\alpha$, we have $f_i(\mathbf{x}_e)=0,
i=1,\ldots, N,$ and starting from an arbitrary $\mathbf{x}$, the
system (\ref{eqnset}) will converge to an equilibrium solution
$\mathbf{x}_e$. Equation~\ref{eqnset} can be written compactly as
\begin{equation}
\dot{\mathbf{x}} = \mathbf{f}(\mathbf{x}). \label{eqncompact}
\end{equation}
where $\mathbf{f}=(f_1,f_2,\ldots,f_N)$. Here $\mathbf{f}$, the
ordered $N\text{-tuple}$, is the consensus protocol and we are
interested in protocols which will take the dynamical system
(\ref{eqnset}) to consensus.

It is possible to associate a graph with every dynamical system of
the form (\ref{eqnset}). Let $\mathcal{G}=(\mathcal{V},\mathcal{E})$
be a such graph with vertices $\mathcal{V}$ and edges $\mathcal{E}$.
Then $\mathcal{V}=\{1,2,\ldots,N\}$ and a directed edge
$e_{ij}=(i,j)$ from node $i$ to node $j$ belongs to $\mathcal{E}$ if
$f_j$ is a function of $x_i$. If $f_j$ is a function of $x_i$, we
say that $x_i$ belongs to the support of $f_j$, that is, $x_i\in
supp(f_j)$. Following Stigler~\cite{stigler}, we call the graph
$\mathcal{G}$ the {\bf dependency graph} of (\ref{eqnset}). The
neighbors of an agent/node $x_i$ are all those nodes which are in
the support of $f_i$. Thus, the set of all neighbors of a node $i$
is
\begin{equation}
\mathcal{N}_i = \{j: e_{ji} \in \mathcal{E}\}. \label{eqnneighbor}
\end{equation}
The most widely analyzed protocol for consensus \cite{olfati1} is
a linear protocol of form
\begin{equation} \dot{x}_i = \sum_{j \in \mathcal{N}_i} (x_j - x_i),
\quad i=1,\ldots,N. \label{linprotocol}
\end{equation}
The {\bf adjacency matrix}, $A=[a_{ij}]$, associated with a graph
$\mathcal{G}=(\mathcal{V},\mathcal{E})$ is defined as
\[ a_{ij} = \left\{
\begin{array}{ll}
1 & \text{if } e_{ij}\in\mathcal{E} \\
0 & \text{otherwise.}
\end{array}
\right. \] A diagonal matrix called the {\bf degree matrix},
$D=[d_{ii}]$, is defined as $d_{ii}=$ number of edges directed
towards $x_i$, which is equal to ${\sum_{j=1}^N a_{ij}}$. The {\bf
graph Laplacian}, $L$, is defined as $L=D-A$. Using the definition
of graph Laplacian, we can rewrite (\ref{linprotocol}) as
\begin{equation}
\dot{\mathbf{x}}= - L \mathbf{x}
\end{equation}
There is a great amount of literature on consensus under linear
protocols. The assignment of a graph structure to the problem makes
the problem amenable to application of tools from algebraic graph
theory like graph Laplacian~\cite{olfati2}, stochasticity of
nonnegative adjacency matrices~\cite{ali}, etc. to analyze the
consensus in system given by (\ref{linprotocol}) or a normalized
form of it. However, these methods for analyzing stability and
convergence to consensus of the system of equations may not be
useful for nonlinear protocols.

There has been a few papers, in consensus literature, addressing
consensus under nonlinear protocols. Olfati-Saber and
Murray~\cite{olfati2} considers consensus under a nonlinear protocol
of form
\begin{equation}
\dot{x}_i = \sum_{j\in \mathcal{N}_i} h_i(x_j-x_i), \quad
i=1,\ldots,n \label{olfatiproto}
\end{equation}
where $h_i, i=1,\ldots,n$ are functions which are uneven, locally
Lipschitz, and strictly increasing.

Liu and Chen~\cite{liu} extends this to protocols of form
\begin{equation}
\dot{x}_i = \sum_{j\in \mathcal{N}_i} \left(h(x_j)-h(x_i)\right),
\quad i=1,\ldots,N \label{Introduction:Equation:LiuProtocol}
\end{equation}
where $h$ is assumed to be an increasing function. However, this
protocol demands that $h$ be same for all agents.

Moreau~\cite{moreau} proves the convergence of consensus protocol
under the assumption that the protocol is such that, in the
discretized version of it, the updated value of a node is a strict
convex combination of current values of the node and its
neighbors . This implies that the continuous time protocol of the
form
\begin{equation}
\dot{x}_i = \sum_{j\in \mathcal{N}_i} h_j(\mathbf{x})(x_j-x_i),
\quad i=1,\ldots,N \label{Introduction:Equation:MoreauProtocol}
\end{equation}
will achieve consensus.

This paper we pose problem of consensus under polynomial protocols
in the framework of 
computational algebra and give necessary conditions for the convergence. 
This paper is organized as follows. In
\S~\ref{PreliminaryObservations} we give preliminary observations
on polynomial consensus and review basic background in computational
algebra that is required. We present main results in
\S~\ref{NecessaryConditions} and give concluding remarks in
\S~\ref{Conclusions}

\section{Preliminary Observations and Basic Computational Algebra}
\label{PreliminaryObservations}
%-----------------------------------------------
\subsection{Preliminary Observations}
We give few necessary conditions for convergence to a consensus
under polynomial protocols, using the tools and language of
algebraic geometry. It is clear that for a consensus to be achieved,
the protocol $\mathbf{f}$ should use all the nodal values. A
consensus is not achieved, except for some particular initial
conditions, if values of one or mode nodes are not used in the
dynamics (\ref{eqncompact}). This leads us to the following
proposition.
\begin{proposition}
{\label{Consensus:Proposition:PolynomialProtocol:NecessaryConditionSupport}
If the system $\dot{x}_i=f_i,i=1,\ldots,N$, achieves consensus, then
for all $j\in\{i,\ldots,N\}$, there exists $i\in\{i,\ldots,N\}$ such
that $x_j\in supp(f_i)$.}
\end{proposition}
The proposition asserts that the dependency graph of a protocol that
achieve consensus should not have any isolated nodes.

Also, if more than one nodal values are not being updated, it is not
possible to arrive at a consensus in general. However, there might
exist a few particular initial conditions for which a consensus on
nodal values is reached. For example, the initial values of the
nodes whose values are not being updated are equal and equal to the
consensus value achieved by other nodes. However, this is an
exception rather than a rule. Therefore, we have the following
necessary condition.
\begin{proposition}
{\label{Consensus:Proposition:PolynomialProtocol:NecessaryConditionZeroPolynomial}
If more than one $f_i$ are zero polynomials, then the system
$\dot{x}_i=f_i,i=1,\ldots,N$, will not achieve consensus.}
\end{proposition}
A result from theory of linear consensus protocols states that a
consensus is not achieved if the underlying graph does not have a
directed spanning tree \cite{ren}.
Proposition~\ref{Consensus:Proposition:PolynomialProtocol:NecessaryConditionZeroPolynomial}
has the same flavor.

From here onwards when we say that a protocol $\mathbf{f}$ satisfies
the conditions for
Propositions~\ref{Consensus:Proposition:PolynomialProtocol:NecessaryConditionSupport}
and
\ref{Consensus:Proposition:PolynomialProtocol:NecessaryConditionZeroPolynomial},
we mean that the protocol involves all the nodal values and at most
one of the polynomials in $\mathbf{f}$ is a zero polynomial. And,
whenever one of the polynomials is a zero polynomial, by the set
$\{f_1,\ldots,f_N\}$ we mean the set without the zero polynomial.

If the system of equations in Eq.~(\ref{eqncompact}) leads to a
consensus, then we have $\dot{\mathbf{x}}=0$, which implies
$\mathbf{f}(\mathbf{x})=0$, asymptotically. The equilibrium points
of the system of equations in Eq.~(\ref{eqncompact}) are the `roots'
of the set of polynomial equations $\mathbf{f}(\mathbf{x})=0$. Now
we need some notions of algebraic geometry and Gr{\"{o}}bner bases.

%----------------------------------------
\subsection{Basics of Algebraic Geometry and Gr\"{o}bner bases}
       Throughout this paper, $\Bbbk$ represents a field (\eg
       $\mathbb{R}$, $\mathbb{C}$).
       %The algebraic closure of the field $k$ is
       %represented by $\bar{k}$ (the algebraic closure of $\mathbb{R}$ is
       %$\mathbb{C}$).
%       From now on, $k$ represents the field $\mathbb{R}$ and $\bar{k}$
%       represents $\mathbb{C}$.
       Set of all monomials in indeterminates $x_{1},\ldots, x_{N}$ is
       denoted by $\SetOfAllMonomials$ and set of all
       polynomials in indeterminates $x_{1},\ldots, x_{N}$
       with coefficients in $\Bbbk$ is denoted by $\PolRing$. Let
       $f_{1},\ldots,f_{s} \in \PolRing$.
       We use the notation $\mathcal{V}(f_{1}, \ldots, f_{s})
       =V$ to represent the varieties, where
       \begin{displaymath}
         V = \{(c_{1},\ldots c_{n}) \in \Bbbk^{N} : f_{i}(c_{1},\ldots
         c_{N}) =0, 1 \leq i \leq s \} \enspace.
       \end{displaymath}
       $V$ is uniquely
       determined by the ideal
       generated by $f_{1}, \ldots, f_{s}$. This ideal is denoted by ${\langle
       f_{1}, \ldots, f_{s}  \rangle}$ and hence we have
       \begin{displaymath}
     {\mathcal{V}}(f_{1}, \ldots, f_{s}) =
       {\mathcal{V}} ({\langle
       f_{1}, \ldots, f_{s}  \rangle} ) \enspace.
       \end{displaymath}

%\begin{definition}
%{\label{Consensus:Definition:PolynomialProtocol:AffineVariety}       
%  Let
%$f_1,\ldots,f_s$ be polynomials in $\mathbb{K}[x_1,\ldots,x_N]$.
%Then we set
%\[\mathcal{V}(f_1,\ldots,f_s) = \left\{(a_1,\ldots,a_N)\in \mathbb{K}^N:
%f_i(a_1,\ldots,a_N)=0 \text{ for all } 1\leq i \leq s \right\}.\] We
%call $\mathcal{V}(f_1,\ldots,f_s)$ the {\bf affine variety} defined
%by $f_1,\ldots,f_s$.}
%\end{definition}
The set of equilibrium points of most of the consensus protocols
%(for example, all the protocols considered in
%Examples~\ref{Consensus:Figure:PolynomialProtocol:Example1}--\ref{Consensus:Figure:PolynomialProtocol:Example5})
contain the subspace $x_1=\cdots = x_N$. Let us call this subspace
where states of all agents are equal as $\mathcal{S}$. If a protocol
$\mathbf{f}$ leads to consensus, then it should have at least one
equilibrium point belonging to the subspace $\mathcal{S}$. Thus we
get a necessary condition for convergence of Eq.~(\ref{eqncompact})
as follows. If Eq.~(\ref{eqncompact}) converges, then
$\mathcal{S}\cap \mathcal{V}(\mathbf{f})$ is nonempty.
Geometrically, the subspace $\mathcal{S}\subset \Bbbk^N$ is a
solution to set of equations
\begin{equation}
\mathbf{s} = \{x_1-x_2,x_2-x_3,\ldots,x_{N-1}-x_N\}
\label{Consensus:Equation:PolynomialProtocol:Subspace}
\end{equation}
In the language of algebraic geometry, $\mathcal{S}$ is the variety
of $\mathbf{s}$, that is, $\mathcal{S} = \mathcal{V}(\mathbf{s}).$
Therefore, a consensus would imply $\mathcal{V}(\mathbf{s})\cap
\mathcal{V}(\mathbf{f})$ is non-empty.

%We now introduce the notion of ideals and the ideal generated by a
%set of polynomials.
%\begin{definition}
%{\label{Consensus:Definition:PolynomialProtocol:Ideal} A subset $I
%\subset \mathbb{K}[x_1,\ldots,x_N]$ is an {\bf ideal} if it
%satisfies:
%\begin{enumerate}
%\item $0 \in I$.
%\item If $f,g \in I$, then $f+g \in I$.
%\item If $f \in I$ and $h \in \mathbb{K}[x_1,\ldots,x_N]$, then $hf \in I$.
%\end{enumerate}
%}
%\end{definition}
%% Here $0$ denotes the zero polynomial. The ideal generated by a set
%% of polynomials is the smallest ideal containing them and is defined
%% as follows.
%% \begin{definition}
%% {\label{Consensus:Definition:PolynomialProtocol:IdealGeneratedBy}Let
%% $f_1,\ldots,f_s$ be polynomials in $\Bbbk[x_1,\ldots,x_N]$.
%% Then we set
%% \[ \langle f_1,\ldots,f_s \rangle = \left\{\sum_{i=1}^s h_if_i: h_1,\ldots,h_s
%% \in \Bbbk[x_1,\ldots,x_N] \right\}.\] It turns out that
%% $\langle f_1,\ldots,f_s \rangle$ is an ideal and we call it the {\bf
%% ideal generated by} $f_1,\ldots,f_s$.}
%% \end{definition}
Let ${J}$ be the ideal generated by set of polynomials $\mathbf{s}$
describing the subspace $\mathcal{S}$ as described before. It is
easy to see that the affine variety of the set of polynomials
$\mathbf{s}$ is equal to the affine variety of the ideal ${J}$
generated by them \cite{cox}, that is, $\mathcal{V}({J})=
\mathcal{V}(\mathbf{s})=\mathcal{S}$.

%\noindent{\bf Remark}
Since $\mathcal{S}$ is an irreducible variety (cannot be written as
the union of non-empty varieties), the ideal ${J}$ is prime
\cite{cox}. Since $J$ is prime, the radical ideal\footnote{The
radical ideal of $J$, denoted by $\sqrt{J}$, is the set
$\{f\in\Bbbk[x_1,\ldots,x_N]:f^n\in J \text{ for some } n \in
\mathbb{Z}_{> 0}\}$} of ${J}$ is itself. Now, by Hilbert's Strong
Nullstellansatz \cite{cox}, the ideal $\mathcal{I}(\mathcal{S})$,
the ideal of all polynomials that vanishes at every point on
$\mathcal{S}$, is $J$.

Throughout this article, we will be using some results from
algebraic geometry. These will be given as propositions and theorems
without proofs and we refer an interested reader for details to Cox
et al.~\cite{cox}.
\begin{theorem}
{\label{Consensus:Theorem:PolynomialProtocol:SumofIdeals} \cite{cox}
If $I$ and $J$ are ideals in $\Bbbk[x_1,\ldots,x_N]$, then
$\mathcal{V}(I+J)=\mathcal{V}(I)\cap \mathcal{V}(J)$.}
\end{theorem}
 Here, $I+J$
is defined as follows.
\begin{definition}
{\label{Consensus:Definition:PolynomialProtocol:SumofIdeals} If $I$
and $J$ are ideals of ring $\Bbbk[x_1,\ldots,x_N]$, then the
{\bf sum} of $I$ and $J$, denoted by $I+J$, is the set
\[I+J = \{f+g|f\in I \;\text{and}\; g\in J\}. \]}
\end{definition}
In fact, $I+J$ is the smallest ideal containing both $I$ and $J$
\cite{cox}. We have the following proposition.
\begin{proposition}
{
\label{Consensus:Proposition:PolynomialProtocol:SumofIdealsVarietyNonEmptyforConsensus}
Let ${I}$ be the ideal generated by $\{f_1,\ldots,f_N\}, f_i \in
\Bbbk[x_1,\ldots,x_N]$, $i = 1,\ldots,N$ and ${J}$ be the ideal
$\langle x_1-x_2,x_2-x_3,\ldots,x_{N-1}-x_N \rangle$. Let
$\mathcal{S} \subset \Bbbk^N$ be the subspace given as
$\mathcal{S}=\{
(x_1,\ldots,x_N)\in\Bbbk^N:x_1=x_2=\cdots=x_N\}$. Then, if the
polynomial dynamical system $\dot{x}_i = f_i, i = 1,\ldots,N$
attains consensus, then $\mathcal{V}({I}+{J})$ is not empty and
${I}+{J}$ is a proper subset of $\Bbbk[x_1,\ldots,x_N]$. }
\end{proposition}
\begin{proof}
{ A necessary condition for consensus is that at least one of the
stationary points of the system $\dot{x}_i = f_i, i = 1,\ldots,n$
should belong to $\mathcal{S}$. Set of stationary points or the
invariant set of this system is the solution to set of polynomial
equations $f_1=0,\ldots,f_N=0$ which is given by $\mathcal{V}(I)$.
Thus the necessary condition demands that $\mathcal{V}(I) \cap
\mathcal{S}$ is not empty. Since $\mathcal{S}$ is Zariski
closed\footnote{A set belonging to $\Bbbk^N$ is said to be
Zariski closed if it is the solution to a set of polynomials in
$\Bbbk[x_1,\ldots,x_N]$.} as $\mathcal{S}$ is the variety of
ideal $\mathcal{J}$, we get $\mathcal{V}(I) \cap
\mathcal{V}({J})\neq \emptyset$. This along with
Theorem~\ref{Consensus:Theorem:PolynomialProtocol:SumofIdeals} gives
that $\mathcal{V}({I}+{J})$ is not empty. Since
$\mathcal{V}({I}+{J})$ is not empty, by Hilbert's Weak
Nullstellensatz \cite{cox}, ${I}+{J}$ is proper ideal of
$\Bbbk[x_1,\ldots,x_N]$.}
\end{proof}
The statements on ideals can be checked by calculating their
Gr\"{o}bner bases which can be done using any symbolic package that
supports algebraic geometric calculations, given the set of
polynomials that generate the ideal. Now, we define a Gr\"{o}bner
basis.

Gr{\"{o}}bner bases is a
    generalization of the division
    algorithm in a single variable case ($\Bbbk[x]$) to the multivariate case
    ($\Bbbk[x_{1},\ldots,x_{N}]$). Gr\"{o}bner basis also generalizes Gaussian elimination
    in linear polynomials to nonlinear polynomials. For the multivariate division algorithm,
    we need the notion of monomial order.
       %---Definition:Monomial Ordering
       \begin{definition}
       {
       A {\bf monomial order} or {\bf term order} on $k[x_{1},\ldots,x_{N}]$ is
       a relation $\prec$ on ${\mathbb{Z}}_{\geq 0}^{N}$ that  satisfies
       following conditions (i)~$\prec$ is a total ordering on
       ${\mathbb{Z}}_{\geq 0}^{N}$, (ii)~if $\alpha \prec \beta$, for
       $\alpha, \beta \in {\mathbb{Z}}_{\geq 0}^{N}$ then for any $\gamma
       \in {\mathbb{Z}}_{\geq 0}^{N}$ it holds $\alpha + \gamma \prec \beta
       + \gamma$, and (iii) $\prec$ is a well-ordering on
       ${\mathbb{Z}}_{\geq 0}^{N}$.}
       \end{definition}
       %---EndDefinition
       Given such an ordering $\prec$, one can define the leading term
       of non-zero polynomial $f \in \Bbbk[x_{1},\ldots,x_{N}]$ as a term of $f$ (the
       coefficient times its monomial) whose monomial is maximal for
       $\prec$. We denote this leading term by $\text{LT}_{\prec}(f)$ and
       the corresponding monomial by $\text{LM}_{\prec}(f)$.
       %---Definition Monomial Ideal
       \begin{definition}{
    An ideal $\mathfrak{a} \subset \Bbbk[x_{1},\ldots,x_{N}]$ is said to be a
    monomial
    ideal if there is a set $A \subset {\mathbb{Z}}_{\geq 0}^{N}$,
    possibly
    infinite, such that $\mathfrak{a} = \langle x^{\alpha}: \alpha
    \in
    A \rangle$.}
       \end{definition}
       Given any ideal $\mathfrak{a} \subset \Bbbk[x_{1},\ldots,x_{N}]$, the ideal
       defined as $\langle \text{LM}_{\prec}(f): f \in \mathfrak{a} \rangle$
       is a monomial ideal and is denoted by $\text{LM}_{\prec}(\mathfrak{a})$,
       which is known as leading monomial ideal of $\mathfrak{a}$. By
       Dickson's
       lemma~\cite[p.~69]{cox},
        % or by Noetherianity,
         the ideal
       $\text{LM}_{\prec}(\mathfrak{a})$ is generated by a finite set of
       monomials. Dickson's lemma and the multivariate division
       algorithm leads to a proof of Hilbert bases theorem which states
       that every
       polynomial ideal can be finitely generated, which further lead
       to a definition of Gr{\"{o}}bner
       basis~\cite[\S~2.5]{cox}.
       %---Definition: Grobner bases
       \begin{definition}{
        Fix a monomial order $\prec$ on $\Bbbk[x_{1},\ldots,x_{N}]$. A finite subset $G =
        \{g_{1},\ldots,g_{s}\}$ of an ideal $\mathfrak{a} \subset
        \Bbbk[x_{1},\ldots,x_{N}]$ is a Gr{\"{o}}bner basis if and only if
        $
           \text{LM}_{\prec}(\mathfrak{a})$ $=$ $\langle \text{LM}_{\prec}(g_{1}),
           \ldots,$
           $\text{LM}_{\prec}(g_{s})  \rangle.
        $}
       \end{definition}
       Given a set of generators of an ideal, the Buchberger
       algorithm~\cite{buchberger}
       can be used to compute a Gr{\"{o}}bner basis of the ideal with respect
       to various term orders. The algorithm and its variants are
       implemented in most symbolic computation programs. Note that
       a Gr{\"{o}}bner basis is not unique, but one can transform it to
       a reduced
       Gr{\"{o}}bner basis which is unique for every ideal in
       $\Bbbk[x_{1},\ldots,x_{N}]$. In the sequel, when we say `the
       Gr\"{o}bner basis', we mean the reduced Gr\"{o}bner basis.
       The Buchberger
       algorithm provides a common generalization of the
       Euclidean division algorithm and the Gaussian elimination algorithm to
       multivariate polynomial rings.

%=================Section: Necessary Conditions==================
\section{Consensus under polynomial protocols: Necessary conditions}
\label{NecessaryConditions}
Now, an equivalent statement of
Proposition~\ref{Consensus:Proposition:PolynomialProtocol:SumofIdealsVarietyNonEmptyforConsensus}
in terms of Gr\"{o}bner basis can be given as follows.
\begin{corollary}
{\label{Consensus:Corollary:PolynomialConsensus:NonEmptyGrobnerBasis}
If ${G}$ is a Gr\"{o}bner basis of ${I}+{J}$ where $I$ and $J$ are
as defined in
Proposition~\ref{Consensus:Proposition:PolynomialProtocol:SumofIdealsVarietyNonEmptyforConsensus},
then the affine variety of ${G}$ is not empty and the ideal
generated by ${G}$ is a proper ideal of
$\Bbbk[x_1,\ldots,x_N]$.}
\end{corollary}
%\noindent {\bf Remarks}
%\begin{quote}
The above condition will always be satisfied if none of the $f_i, i
= 1,\ldots,N$, has a constant term, as then, $\mathbf{x}=\mathbf{0}$
is a common solution of polynomials in both $I$ and $J$. We call
this the {\it trivial consensus} where a consensus occurs if the
initial conditions on all the nodes are simultaneously zeros.
Similarly, consensus occurs for other special initial cases too.
Since the conditions of
Proposition~\ref{Consensus:Proposition:PolynomialProtocol:SumofIdealsVarietyNonEmptyforConsensus}
also holds for protocols that can achieve trivial consensus also, it
is a weak result.
%\end{quote}

From
Proposition~\ref{Consensus:Proposition:PolynomialProtocol:SumofIdealsVarietyNonEmptyforConsensus}
we have that, to show that a system does not attain consensus, it is
enough to show that $I+J$ is the entire polynomial ring
$\Bbbk[x_1,\ldots,x_N]$, or equivalently by Hilbert's Weak
Nullstellansatz, the Gr\"{o}bner basis of ideal $I+J$ is $1$ (since
$1$ generates the whole ring $\Bbbk[x_1,\ldots,x_N]$ which has
an empty variety). For this, we need a way to calculate the ideal
$I+J$ given $I$ and $J$, which the following proposition gives.
\begin{proposition}
{\label{Consensus:Proposition:PolynomialProtocol:SumOfIdealsIsIdeal}
\cite{cox}If $I=\langle f_1,\ldots,f_r \rangle$ and $J=\langle
g_1,\ldots,g_s \rangle$, then $I+J = \langle
f_1,\ldots,f_r,g_1,\ldots,g_s \rangle$.}
\end{proposition}
Thus to check the
Proposition~\ref{Consensus:Proposition:PolynomialProtocol:SumofIdealsVarietyNonEmptyforConsensus},
given $\langle\mathbf{f}\rangle$, the ideal generated by protocols,
and $\langle\mathbf{s}\rangle$, the ideal the variety of which is
$\mathcal{S}$, we need to calculate the variety of ideal
$\langle\mathbf{f}\rangle+\langle\mathbf{s}\rangle$ which by
Proposition~\ref{Consensus:Proposition:PolynomialProtocol:SumOfIdealsIsIdeal}
is equal to the variety of $\langle\mathbf{f},\mathbf{s}\rangle$.
Since a Gr\"{o}bner basis of $\langle\mathbf{f},\mathbf{s}\rangle$
has the same variety as that of
$\langle\mathbf{f},\mathbf{s}\rangle$, it is enough to calculate the
Gr\"{o}bner basis of $\langle\mathbf{f},\mathbf{s}\rangle$ and look
at its variety. This can be done using any software that support
computation of Gr\"{o}bner basis. For the examples given in this
chapter, we use Mathematica.

%% We apply
%% Corollary~\ref{Consensus:Corollary:PolynomialConsensus:NonEmptyGrobnerBasis}
%% and
%% Proposition~\ref{Consensus:Proposition:PolynomialProtocol:SumOfIdealsIsIdeal}
%% to the Example~\ref{Consensus:Example:PolynomialProtocol:Example3}.
%% For Example~\ref{Consensus:Example:PolynomialProtocol:Example3}, the
%% Gr\"{o}bner basis is obtained as follows. Given below are the input
%% to and output from Mathematica. The input is typeset in boldface and
%% we follow this convention throughout this chapter.

%% \noindent\(\pmb{\text{GroebnerBasis}[\{\{x_2x_5-x_1^2,x_1x_3x_4-x_2^3,x_5-x_3,x_2-x_4,x_1x_3-x_5^2\},}\\
%% \pmb{\{x_1-x_2,x_2-x_3,x_3-x_4,x_4-x_5\}\}]}\)\\
%% %,\{x_1,x_2,x_3,x_4,x_5\}
%% \noindent\(\{x_4-x_5,x_3-x_5,x_2-x_5,x_1-x_5\}\)

%% We see that the variety of the Gr\"{o}bner basis is not empty and
%% thus the protocol of
%% Example~\ref{Consensus:Example:PolynomialProtocol:Example3}
%% satisfies the necessary condition for consensus required by
%% Proposition~\ref{Consensus:Proposition:PolynomialProtocol:SumofIdealsVarietyNonEmptyforConsensus}.
%% Now, we give an example of a protocol where consensus is not
%% achieved and show its Gr\"{o}bner basis to be $\{1\}$, which the
%% Proposition~\ref{Consensus:Proposition:PolynomialProtocol:SumofIdealsVarietyNonEmptyforConsensus}
%% insists.
\begin{example}
{\label{Consensus:Example:PolynomialProtocol:BiasNonlinearity} We
consider a nonlinear protocol as follows.
\begin{displaymath}
  \begin{array}{rrr}
\dot{x}_1 = x_1+x_2+1,\:\:\:&
\dot{x}_2 = x_1+x_3+3,\:\:\:& 
\dot{x}_3 = x_2-x_3.
\end{array}
\end{displaymath}
%\begin{eqnarray*}
%\dot{x}_1 &=& x_1+x_2+1 \notag\\
%\dot{x}_2 &=& x_1+x_3+3  \label{Consensus:Equation:PolynomialProtocol:BiasNonlinearity}\\
%\dot{x}_3 &=& x_2-x_3 \notag
%\end{eqnarray*}
\noindent\(\pmb{\text{GroebnerBasis}[\{\{x_1+x_2+1,x_1+x_3+3,x_2-x_3\},}
\pmb{\{x_1-x_2,x_2-x_3\}\}]}\)

\noindent\(\{1\}\) }
\end{example}

%\noindent {\bf Remarks}
%\begin{quote}
For the system in Eq.~\eqref{eqnset}, a consensus is achieved if
there exists an $\alpha\in\mathbb{R}$ such that
$x_1=x_2=\cdots=x_N=\alpha$ is an equilibrium point. However,
systems for which $x_1=x_2=\cdots=x_N=\alpha$ is an equilibrium
point for all $\alpha\in\mathbb{R}$ are of particular interest in
the consensus community. Therefore, in the rest of the paper, we
consider only such protocols. The necessary condition as given in
Proposition~\ref{Consensus:Proposition:PolynomialProtocol:SumofIdealsVarietyNonEmptyforConsensus}
then amounts to $\mathcal{V}(J)\subseteq\mathcal{V}(I)$, where
$I=\langle f_1,f_2,\ldots,f_N\rangle$ and $J=\langle
x_2-x_1,x_3-x_2,\ldots,x_{N-1}-x_N\rangle$.
%\end{quote}

A consensus cannot be achieved over a graph that is disconnected.
Thus a necessary condition for a consensus to occur under a protocol
is that its dependency graph should be connected. In the case of a
linear protocol with a dependency graph that is directed, it should
be strongly connected (or at least should have a directed spanning
tree) for consensus to occur \cite{ren}. A result from algebraic
graph theory states that if a directed graph is strongly connected,
then for a linear consensus protocol the graph Laplacian matrix
associated with it is irreducible \cite{godsil}. This essentially
means that by a permutation or an automorphism of the variables, the
associated graph Laplacian can be made to be of a block diagonal
form. The number of blocks are equal to the number of disconnected
components of the graph \cite{ren}. Given a protocol $\mathbf{f} =
\{f_1,\ldots,f_N\}$, its dependency graph
$\mathcal{G}=(\mathcal{V},\mathcal{E})$ can be found. It is possible
to characterize the properties of this graph by defining the
following relation.
\begin{definition}{
\label{Consensus:Definition:PolynomialProtocol:PartialOrderBetweenNodes}
Given two nodes $i,j\in \mathcal{V}$, we say $i\prec_p j$ if there
exists a path from node $j$ to node $i$. If $i\prec_p j$ and
$j\prec_p i$, then nodes $i$ and $j$ are path-equivalent and we
denote this by $i\sim_p j$.}
\end{definition}
The relation $\prec_p$ is anti-symmetric since if $i\prec_p j$ and
$j\prec_p i$, then $i\sim_p j$. It is transitive. The above relation
is reflexive and hence a partial order if the dependency graph
contains self-loops.

We have the following results that are immediate from the definition
of the order relation $\prec_p$.
\begin{proposition}{
\label{Consensus:Proposition:PolynomialProtocol:GraphCharacterizationWithOrderRelation}
Let $\mathcal{G}=(\mathcal{V},\mathcal{E})$ be graph with
$\mathcal{V} = \{1,\ldots,N\}$. Let
$\text{Maximal}_{\prec_p}\mathcal{V}$ denote the set of maximal
elements in $\mathcal{V}$ under the order relation $\prec_p$. Then,
\begin{enumerate}
\item The graph $\mathcal{G}$ is strongly connected if and only if $\text{Maximal}_{\prec_p}\mathcal{V}=\mathcal{V}$
\item The graph $\mathcal{G}$ has a directed spanning tree if and
only if \\
$\#\left(\text{Maximal}_{\prec_p}\mathcal{V}/\sim_p\right)=1$
%\item The number of disconnected components of $\mathcal{G}$ is
%equal to \\
%$\#\left(\text{Maximal}_{\prec_p}\mathcal{V}/\sim_p\right)-1$
\end{enumerate}
where $\#$ denotes the cardinality of a set.}
\end{proposition}
The necessary and sufficient condition for achieving a consensus
under linear protocol over a graph is the existence of a directed
spanning tree \cite{ren}. In fact, the existence of a directed
spanning tree is a necessary condition not only for the linear
consensus protocol but also for any protocol that achieves consensus
over a graph.
%% We give examples of consensus under polynomial
%% protocols whose dependency graphs are not strongly connected but
%% have directed spanning trees. Towards this, we consider protocols of
%% Examples~\ref{Consensus:Example:PolynomialProtocol:Example1}-\ref{Consensus:Example:PolynomialProtocol:Example3}
%% over a graph as shown in
%% Fig.~\ref{Consensus:Figure:PolynomialProtocol:DependencyGraphDirectedSpanningTree}.
\begin{figure}[htbp]
\centering
\includegraphics[scale=.3]{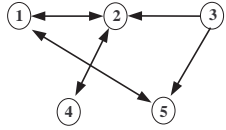}
\caption{Dependency graph with a directed spanning tree rooted at
node 3.}
\label{Consensus:Figure:PolynomialProtocol:DependencyGraphDirectedSpanningTree}
\end{figure}
\begin{figure}[htbp]
\centering {
\includegraphics[scale=.4]{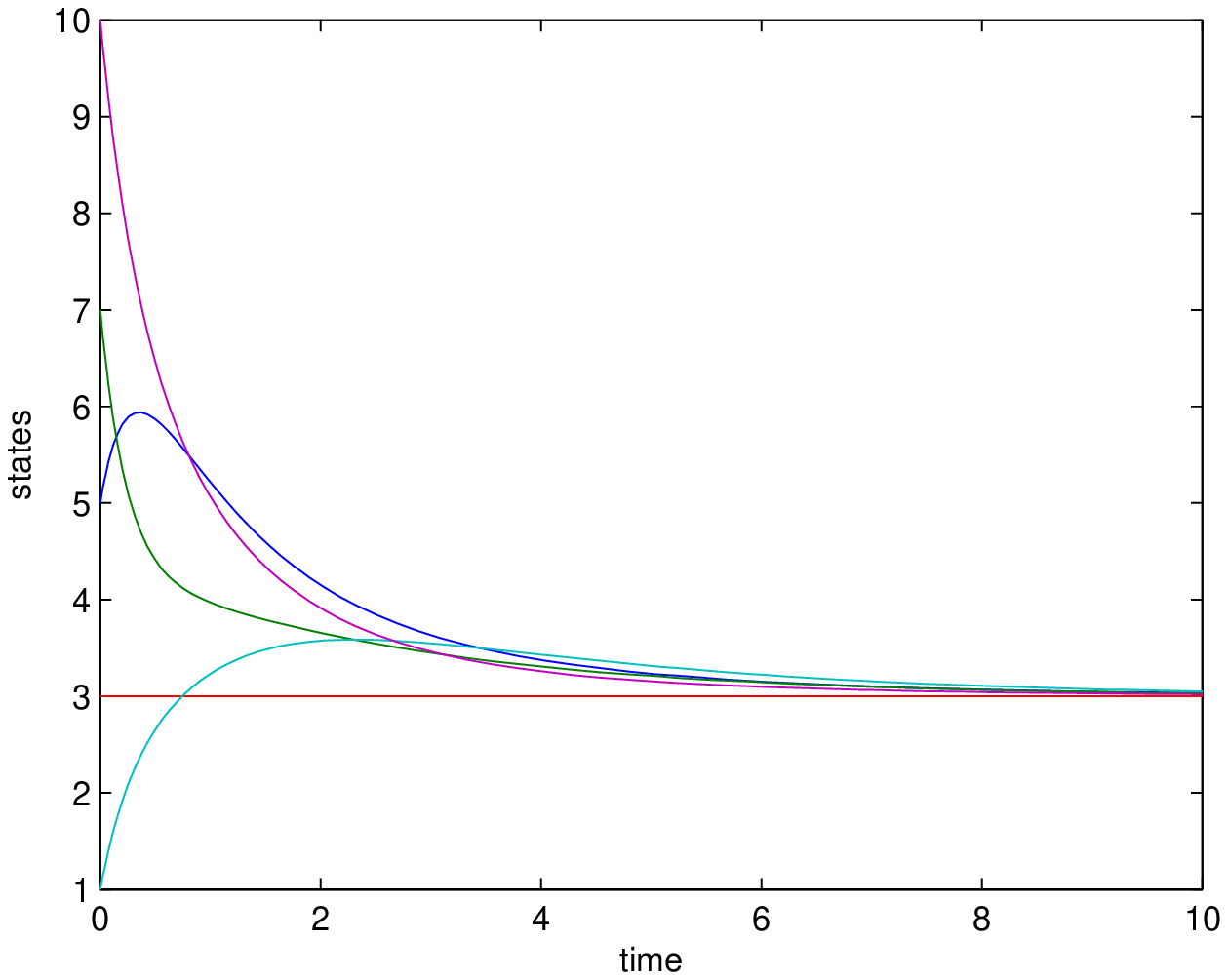}
\caption{Consensus under protocol in
Eq.~\eqref{Consnesus:Equation:PolynomialProtocol:LinearProtocol:DirectedSpanningTree}.}
%with a dependency graph in
%Fig.\ref{Consensus:Figure:PolynomialProtocol:DependencyGraphDirectedSpanningTree}.}
\label{Consensus:Figure:PolynomialProtocol:DependencyGraphDirectedSpanningTree:Example4}
}
\end{figure}

\begin{example}
{
\label{Consnesus:Example:PolynomialProtocol:LinearProtocol:DirectedSpanningTree}
The linear protocol
\begin{equation}
\dot{x}_i=\sum_{j\in\mathcal{N}_i} x_j - {|\mathcal{N}_i|}x_i
\label{Consnesus:Equation:PolynomialProtocol:LinearProtocol:DirectedSpanningTree}
\end{equation}
applied to dependency graph given in
Fig.~\ref{Consensus:Figure:PolynomialProtocol:DependencyGraphDirectedSpanningTree}
is
\begin{displaymath}
  \begin{array}{lllll}
\dot{x}_1 = x_2+x_5 - 2x_1, \:\:\:&
\dot{x}_2 = x_1+x_3+x_4-3 x_2,\:\:\:&
\dot{x}_3 = 0,\:\:\:&\\
\dot{x}_4 = x_2-x_4,\:\:\:&
\dot{x}_5 = x_1+x_3-2 x_5.
\end{array}
\end{displaymath}
%\begin{eqnarray*}
%\dot{x}_1 &=& x_2+x_5 - 2x_1 \\
%\dot{x}_2 &=& x_1+x_3+x_4-3 x_2\\
%\dot{x}_3 &=& 0\\
%\dot{x}_4 &=& x_2-x_4\\
%\dot{x}_5 &=& x_1+x_3-2 x_5
%\end{eqnarray*}
Convergence to consensus under this protocol is given in
Fig.~\ref{Consensus:Figure:PolynomialProtocol:DependencyGraphDirectedSpanningTree:Example4}.
}
\end{example}
\begin{example}
{
\label{Consnesus:Example:PolynomialProtocol:SquaredProtocol:DirectedSpanningTree}
We now consider consensus under a nonlinear protocol
\begin{equation}
\dot{x}_i=\sum_{j\in\mathcal{N}_i} x_j^2 - {|\mathcal{N}_i|}x_i^2
\label{Consnesus:Equation:PolynomialProtocol:SquaredProtocol:DirectedSpanningTree}
\end{equation}
for the dependency graph in
Fig.~\ref{Consensus:Figure:PolynomialProtocol:DependencyGraphDirectedSpanningTree}
\begin{displaymath}
  \begin{array}{lllll}
\dot{x}_1 = x_2^2+x_5^2 - 2x_1^2, \:\:\:&
\dot{x}_2 = x_1^2+x_3^2+x_4^2-3 x_2^2,\:\:\:&
\dot{x}_3 = 0,\:\:\:\\
\dot{x}_4 = x_2^2-x_4^2,\:\:\:&
\dot{x}_5 = x_1^2+x_3^2-2 x_5^2.
\end{array}
\end{displaymath}
%\begin{eqnarray*}
%\dot{x}_1 &=& x_2^2+x_5^2 - 2x_1^2 \\
%\dot{x}_2 &=& x_1^2+x_3^2+x_4^2-3 x_2^2\\
%\dot{x}_3 &=& 0\\
%\dot{x}_4 &=& x_2^2-x_4^2\\
%\dot{x}_5 &=& x_1^2+x_3^2-2 x_5^2
%\end{eqnarray*}
Figure~\ref{Consensus:Figure:PolynomialProtocol:DependencyGraphDirectedSpanningTree:Example5}
shows the convergence to consensus of this protocol. }
\end{example}

\begin{figure}[htbp]
\centering {
\includegraphics[scale=.4]{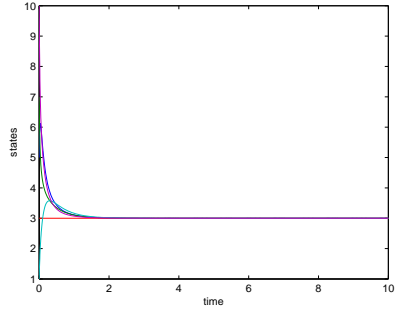}
\caption{Consensus under protocol in
Eq.~\eqref{Consnesus:Equation:PolynomialProtocol:SquaredProtocol:DirectedSpanningTree}.}
%with a dependency graph in
%Fig.\ref{Consensus:Figure:PolynomialProtocol:DependencyGraphDirectedSpanningTree}.}
\label{Consensus:Figure:PolynomialProtocol:DependencyGraphDirectedSpanningTree:Example5}
}
\end{figure}
\begin{example}
 {\label{Consnesus:Example:PolynomialProtocol:ProductProtocol:DirectedSpanningTree} The third protocol considered earlier is
\begin{equation}
\dot{x}_i=\prod_{j\in\mathcal{N}_i} x_j - x_i^{|\mathcal{N}_i|}
\label{Consnesus:Equation:PolynomialProtocol:ProductProtocol:DirectedSpanningTree}
\end{equation}
Applied to dependency graph of
Fig.~\ref{Consensus:Figure:PolynomialProtocol:DependencyGraphDirectedSpanningTree},
we get
\begin{displaymath}
  \begin{array}{rrrrr}
    \dot{x}_1 = x_2x_5 - x_1^2,\:\:\: &
    \dot{x}_2 = x_1x_3x_4- x_2^3,\:\:\: &
    \dot{x}_3 = 0, \:\:\:&
    \dot{x}_4 = x_2-x_4, \:\:\:&
    \dot{x}_5 = x_1x_3- x_5^2.
  \end{array}
\end{displaymath}
%\begin{eqnarray*}
%\dot{x}_1 &=& x_2x_5 - x_1^2 \\
%\dot{x}_2 &=& x_1x_3x_4- x_2^3\\
%\dot{x}_3 &=& 0\\
%\dot{x}_4 &=& x_2-x_4\\
%\dot{x}_5 &=& x_1x_3- x_5^2
%\end{eqnarray*}
Figure~\ref{Consensus:Figure:PolynomialProtocol:DependencyGraphDirectedSpanningTree:Example6}
shows that this protocol leads to consensus. }
\end{example}
\begin{figure}[htbp]
\begin{center}
\includegraphics[scale=.4]{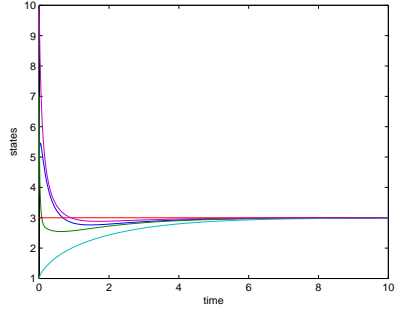}
\end{center}
\caption{Consensus under protocol in
Eq.~\eqref{Consnesus:Equation:PolynomialProtocol:ProductProtocol:DirectedSpanningTree}.}
%with a dependency graph in
%Fig.\ref{Consensus:Figure:PolynomialProtocol:DependencyGraphDirectedSpanningTree}.}
\label{Consensus:Figure:PolynomialProtocol:DependencyGraphDirectedSpanningTree:Example6}\end{figure}
Examples~\ref{Consnesus:Example:PolynomialProtocol:LinearProtocol:DirectedSpanningTree}--\ref{Consnesus:Example:PolynomialProtocol:ProductProtocol:DirectedSpanningTree}
give protocols all of which had a dependency graph that is not
strongly connected but still achieve consensus. This shows that the
dependency graph of a protocol being strongly connected is not a
necessary condition for consensus to occur.

Using the elimination theorem from algebraic geometry, we can get a
 result that enables us to determine whether the dependency graph of a
polynomial protocol has a directed spanning tree. To state the
elimination theorem, we need concepts of elimination order and
elimination ideal as defined below.
          %Definition: Elimination order
      \begin{definition}{\label{Consensus:Definition:PolynomialProtocol:EliminationOrder} 
      Consider $\Bbbk[x_{1},\ldots,x_{n},y_{1},\ldots,y_{m}]$ a
            polynomial ring in indeterminates $x_{1}, \ldots, x_{n}$,
            $y_{1}, \ldots, y_{m}$.
	    %We refer to
            %$\{x_{1},\ldots,x_{n}\}$ as $x$-variables and
            %$\{y_{1},\ldots,y_{m}\}$ as $y$-variables.
	    Let $\prec_{x}$
            and $\prec_{y}$
            be monomial orderings on $x$ and $y$ variables
            respectively. Define an ordering relation $\prec_{[\{x\} \succ \{y\}]}$ on
            ${\mathbb{Z}}_{\geq 0}^{n+m}$ (i.e set of all monomials in
            indeterminates $x_{1},\ldots, x_{n}, y_{1},\ldots, y_{m}$)
            as follows:
        \begin{displaymath}
         x^{a} y^{b}
         \prec_{[\{x\} \succ \{y\}]} x^{c} y^{d}
         \Longleftrightarrow
         \left\{ \begin{array}{l}
                     a\prec_{x} c \\
                 \mathrm{or}\\
                          a = c
                            \:\:\mathrm{and}\:\:
              b \prec_{y} d
            \end{array} \right. \enspace,
        \end{displaymath}
        where $a, c \in {\mathbb{Z}}_{\geq
        0}^{n}$ and $b, d \in {\mathbb{Z}}_{\geq
        0}^{m}$. The term order $\prec_{[\{x\} \succ \{y\}]}$ is called
        {\bf elimination order} with the $x$ variables larger than the $y$
        variables (which is indeed a term order).}
      \end{definition}
\begin{definition}
{\label{Consensus:Definition:PolynomialProtocol:EliminationIdeal}
Given $I=\langle f_1,\ldots,f_s\rangle \subset
\Bbbk[x_1,\ldots,x_N]$ the $l$-th {\bf elimination ideal} $I_l$
is the ideal of $\Bbbk[x_{l+1},\ldots,x_N]$ defined by $I_l = I
\cap \Bbbk[x_{l+1},\ldots,x_N]$}
\end{definition}
Now, we sate the elimination theorem.
\begin{theorem}
{\cite{cox} Let $I\subset \Bbbk[x_1,\ldots,x_N]$ be an ideal
and let $G$ be a Gr\"{o}bner Basis of $I$ with respect to an
elimination order where $\{x_1,x_2,\ldots,x_l\}\succ
\{x_{l+1},\ldots,x_N\}$ for every $0\leq l\leq N$. Then the set
\[G_l = G \cap \Bbbk[x_{l+1},\ldots,x_N]\] is a Gr\"{o}bner Basis of
the $l$-th elimination ideal $I_l$.}
\end{theorem}
If the dependency graph of a linear consensus protocol contains a
directed spanning tree, then the corresponding Laplacian matrix,
$L$, will be of rank $N-1$ \cite{ren}. This means that it is
possible to eliminate $N-2$ variables from the protocol
$L\mathbf{x}$ using Gaussian elimination. Since the elimination in
the multivariate polynomial case is the generalization of the
Gaussian elimination, we expect this nice property of being able to
eliminate $N-2$ variables in the linear protocol case to carry over
to polynomial protocols that attain consensus.
\begin{theorem}
{
\label{Consensus:Theorem:PolynomialProtocol:EliminationIdealNotEmpty}
Let $f_i \in \Bbbk[x_1,\ldots,x_N],$ $i=1,\ldots,N.$ Let
$I=\langle f_1,f_2,\ldots,f_N\rangle$ be such that
$\mathcal{V}(J)\subseteq\mathcal{V}(I)$ for $J=\langle
x_2-x_1,x_3-x_,\ldots,x_{N-1}-x_N\rangle$. If the system $\dot{x}_i
= f_i,i=1,\ldots,N$
  attains consensus, then
  the $(N-2)$-th elimination ideal of $\langle
  f_1,\ldots,f_N\rangle$ is not empty under any elimination order
  with $\{x_{\sigma(1)},\ldots,x_{\sigma(N-2)}\}\prec \{x_{\sigma(N-1)},x_{\sigma(N)}\}$ where
  $\sigma$ belongs to the set of every possible permutation of
  $\{1,\ldots,N\}$.}
\end{theorem}
This follows from the fact that a consensus is achieved only if the
dependency graph has a directed spanning tree. We illustrate this
through the following example.
\begin{example}
{\label{Consensus:Example:PolynomialProtocol:NotPossibleToEliminate}
Consider a polynomial protocol given as
\begin{displaymath}
  \begin{array}{r r r r}
    \dot{x}_1 = x_2 - x_1,\:\:\: &
    \dot{x}_2 = x_1- x_2, \:\:\: &
    \dot{x}_3 = 0, \:\:\:&
    \dot{x}_4 = x_1x_2x_3-x_4^3.
  \end{array}
\end{displaymath}
%\begin{eqnarray*}
%\dot{x}_1 &=& x_2 - x_1 \\
%\dot{x}_2 &=& x_1- x_2\\
%\dot{x}_3 &=& 0\\
%\dot{x}_4 &=& x_1x_2x_3-x_4^3
%\end{eqnarray*}
The corresponding dependency graph is shown in
Fig.~\ref{Consensus:Figure:PolynomialProtocol:NoDirectedSpanningTree}.
This dependency graph does not have a directed spanning tree and
thus will not attain a consensus. Thus by
Theorem~\ref{Consensus:Theorem:PolynomialProtocol:EliminationIdealNotEmpty},
there should exist an elimination order under which it may not
possible to eliminate 2 variables from the Gr\"{o}bner Basis of
$\{x_2 - x_1,x_1- x_2,x_1x_2x_3-x_4^3\}$. If we try to eliminate the
highest two variables under the elimination order
$\{x_4,x_3\}\prec\{x_2,x_1\}$, we get an empty set implying that the
elimination is not possible.

\noindent\(\pmb{\text{GroebnerBasis}[\{x_2-x_1,x_1-x_2,x_1x_2x_3-x_4^3\},\{x_4,x_3\},\{x_2,x_1\}]}\)

\noindent\(\{\}\) }
\end{example}
%----Figure
\begin{figure}[htbp]
\begin{center}
\includegraphics[scale=.3]{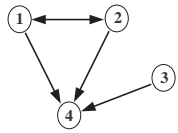}
\end{center}
\caption{Dependency graph of protocol in
Example~\ref{Consensus:Example:PolynomialProtocol:NotPossibleToEliminate}}
\label{Consensus:Figure:PolynomialProtocol:NoDirectedSpanningTree}
\end{figure}

It is possible to extend the result in
Theorem~\ref{Consensus:Theorem:PolynomialProtocol:EliminationIdealNotEmpty},
to the following.
\begin{proposition}
{\label{Consensus:Proposition:PolynomialProtocol:DisconnectedComponents}
  Let $\mathbf{f}=(f_1,\ldots,f_N) \in \Bbbk[x_1,\ldots,x_N]^N$ be such that
$\mathcal{V}(J)\subseteq\mathcal{V}(I)$ for $J=\langle
x_2-x_1,x_3-x2,\ldots,x_{N-1}-x_N\rangle$ and $I=\langle
f_1,f_2,\ldots,f_N\rangle$. Let $l\in\mathbb{N}$ be the smallest $l$
such
  that the $l$-th elimination ideal of $\langle
  f_1,\ldots,f_N\rangle$ is empty under an elimination order with
  $\{x_{i_1},\cdots,x_{i_l}\}\succ \{x_{i_{l+1}},\ldots,x_{i_N}\}$ where $\{i_1,\ldots,i_N\}$ is a
  permutation of $\{1,\ldots,N\}$. Then, $\#\left(\text{Maximal}_{\prec_p}\mathcal{V}/\sim_p\right)=N-l$.
 }
 \end{proposition}
 For
 Example~\ref{Consensus:Example:PolynomialProtocol:NotPossibleToEliminate},
 we have $N=4$, $l=2$, and from Fig.~\ref{Consensus:Figure:PolynomialProtocol:NoDirectedSpanningTree}
 we see that
 \begin{displaymath}
 \#\left(\text{Maximal}_{\prec_p}\mathcal{V}/\sim_p\right)=2.
 \end{displaymath}

\section{Concluding Remarks}
\label{Conclusions}
We explored a novel way of looking at the consensus over
networks--the algebraic geometric way--when the protocols are
polynomials. We gave several necessary conditions for obtaining
consensus under polynomial protocols. This enables one to comment on
the convergence to consensus of a protocol by looking at the
computed Gr\"{o}bner Basis of the set of polynomials in the
protocol. We also gave a sufficient condition for convergence under
polynomial consensus. We conclude by remarking that algebraic
geometry, if used properly, has sufficient powerful tools to analyze
consensus under switching protocols and to design consensus
protocols with desired behaviors.

%\subsubsection*{Acknowledgments.}
%Part of this research was conducted when first author was at
%EURANDOM, Eindhoven.

%\bibliographystyle{splncs}
%\bibliography{papi}

\end{document}